\documentclass{article}

\usepackage{times}  
\usepackage{helvet}  
\usepackage{courier}  
\usepackage{url}  
\usepackage{graphicx}  
\usepackage{mathtools}
\usepackage{booktabs}
\usepackage[utf8]{inputenc}
\usepackage[english]{babel}
\usepackage{amsmath,amssymb,amsfonts}

\usepackage{caption}
\usepackage{graphicx}
\usepackage{subfigure}
\usepackage{amsthm,dsfont}
\usepackage{algorithmic}
\usepackage[ruled, noend]{algorithm2e}
\newcommand{\ignore}[1]{}
\usepackage{color}
\usepackage{authblk}
\usepackage{graphicx}
\usepackage{float} 
\usepackage{multirow}
\usepackage{amsfonts}
\usepackage{comment}
\usepackage{fancyhdr}
\usepackage{rotating}
\usepackage{algorithmic}
\usepackage[ruled, noend]{algorithm2e}
\newtheorem{theorem}{Theorem}[section]

\newtheorem{Lemma}[theorem]{Lemma}
\newtheorem{Instance}[theorem]{Instance}
\newtheorem{Definition}[theorem]{Definition}

\begin{document}
	
	\title{Runtime Analysis of RLS and the (1+1)~EA for the Chance-constrained Knapsack Problem with Correlated Uniform Weights}

\author{Yue Xie \textsuperscript{\rm 1}}
\author{Aneta Neumann \textsuperscript{\rm 1}}
\author{Frank Neumann \textsuperscript{\rm 1}}
\author{Andrew M. Sutton \textsuperscript{\rm 2}}
\affil{ \textsuperscript{\rm 1}Optimisation and Logistics, School of Computer Science,\\ The University of Adelaide, Adelaide, Australia}
\affil{ \textsuperscript{\rm 2}Department of Computer Science,\\ University of Minnesota Duluth, United States of America}
\renewcommand\Authands{ , }
\maketitle
\begin{abstract}
Addressing a complex real-world optimization problem is a challenging task. The chance-constrained knapsack problem with correlated uniform weights plays an important role in the case where dependent stochastic components are considered. We perform a runtime analysis of a randomized search algorithm (RSA) and a basic evolutionary algorithm (EA) for the chance-constrained knapsack problem with correlated uniform weights. We prove bounds for both algorithms for producing a feasible solution. Furthermore, we investigate the behavior of the algorithms and carry out analyses on two settings:  uniform profit value and the setting in which every group shares an arbitrary profit profile. We provide insight into the structure of these problems and show how the weight correlations and the different types of profit profiles influence the runtime behavior of both algorithms in the chance-constrained setting.

\end{abstract}

\section{Introduction}

Evolutionary algorithms are bio-inspired randomized optimization techniques and have been shown to be very successful when applied to various stochastic combinatorial optimization problems \cite{HORNG,TILL,Rakshit}. A significant challenge for real-world applications is that one must often solve large-scale, complex, and uncertain optimization problems where constraint violations have extremely disruptive effects. 

In recent years, evolutionary algorithms for solving dynamic and stochastic combinatorial optimization problems have been theoretically analyzed in a number of articles \cite{He2014,Andrei2017,Vahid2019,Vahid2018}. The techniques that used in runtime analysis has significantly increased understanding of bio-inspired approaches in theoretical field \cite{DROSTE200251,Auger2011,Thomas2013,Frank2012,Carsten2013}. When tackling new problems, such studies typically begin with basic algorithms such as Randomized Local Search (RLS) and (1+1)~EA, which we also investigate in this paper.

An important class of stochastic optimization problems is \emph{chance-constrained} optimization problems \cite{Charnes1959,poojari}. Chance-constrained programming has been carefully studied in the operations research community \cite{Grani2015,Hanasusanto2017,Janiele19}. In this domain, chance constraints are used to model problems and relax them into equivalent nonlinear optimization problems which can then be solved by nonlinear programming solvers \cite{Pu08,Odetayo18,Maryam18}. Despite its attention in operations research, chance-constrained optimization has gained comparatively little attention in the area of evolutionary computation \cite{Liu13}. 

The chance-constrained knapsack problem is a stochastic version of the classical knapsack problem where the weight of the items are stochastic variables. The goal is to maximize the total profit under the constraint that the knapsack capacity bound is violated with a probability of at most a pre-defined tolerance $\alpha$. Recent papers \cite{Yue19,Yue20} study a chance-constrained knapsack problem where the weight of the items are stochastic variables and independent to each other. They introduce the use of suitable probabilistic tools such as Chebyshev's inequality and Chernoff bounds to estimate the probability of violating the constraint of a given solution, providing surrogate functions for the chance constraint, and present single- and multi-objective evolutionary algorithms for the problem. 

The research of chance-constrained optimization problems associated with evolutionary algorithms is an important new research direction from both a theoretical and a practical perspective. Recently, Doerr et al.\ \cite{Benjamin20} analyzed the approximation behavior of greedy algorithms for chance-constrained submodular problems. Assimi et al.~\cite{Hirad20} conducted an empirical investigation on the performance of evolutionary algorithms solving the dynamic chance-constrained knapsack problem. 

From a theoretical perspective, Neumann et al.\ \cite{Aneta20} worked out the first runtime analysis of evolutionary multi-objective algorithms for chance-constrained submodular functions and proved that the multi-objective evolutionary algorithms outperform greedy algorithms. Neumann and Sutton \cite{Frank19} conducted a runtime analysis of the chance-constrained knapsack problem, but only focused on the case of independent weights.

In this paper, analyze the expected optimization time of RLS and the (1+1)~EA on the chance-constrained knapsack problem with correlated uniform weights. This variant partitions the set of items into groups, and pairs of items within the same group have correlated weights. To the best of our knowledge, this is a new direction in the research area of chance-constrained optimization problems. We prove bounds on both the time to find a feasible solution, as well as the time to obtain the optimal solution which has both maximal profit and minimal probability of violating the chance-constrained. In particular, we first prove that a feasible solution can be found by RLS in time bounded by $O(n\log n)$, and by the (1+1)~EA in time bounded by $O(n^2\log n)$. Then, we investigate the optimization time for these algorithms when the profit values are uniform which has been study in the deterministic constrained optimization problems \cite{FRIEDRICH20203}. However, the items in our case are divided into different groups and need to take the number of chosen items from each group into account, and the optimization time bound for RLS becomes $O(n^3)$ and $O(n^3\log n)$ for the (1+1)~EA. 
After that, we consider the more general and complicated case in which profits may be arbitrary as long as each group has the same set of profit values. We show that an upper bound of $O(n^3)$ holds for RLS and $O(n^3 (\log n + \log p_{max}))$ holds for the (1+1)~EA where $p_{max}$ denotes the maximal profit among all items.

This paper is structured as follows. We describe the problem and the surrogate function of the chance constraint in Section \ref{sec:preliminaries} as well as the algorithms. Section \ref{sec:reachfeasibel} presents the runtime results for different algorithms produce a feasible solution, and the expected optimization time for different profit setting of the problem present in Section \ref{sec:optimalSolution} and Section \ref{sec:generalProfit}. Finally, we finish with some conclusions.

\section{Preliminaries}
\label{sec:preliminaries}

The chance-constrained knapsack problem is a constrained combinatorial optimization problem which aims to maximize a profit function and subjects to the probability that the weight exceeds a given bound is no more than an acceptable threshold. In previous research, the weights of items are stochastic and independent of each other. We investigate the chance-constrained knapsack problem in the context of uniform random correlated weights.

Formally, in the chance-constrained knapsack problem with correlated uniform weights, the input is given by a set of $n$ items partitioned to $K$ groups of $m$ items each. We denote as $e_{ij}$ the $j$-th item in group $i$. Each item has an associated stochastic weight. The weights of items in different groups are independent, but the weights of items in the same group $k$ are correlated with one another with a shared covariance $c > 0$, i.e., we have $cov(e_{kj},e_{kl})=c$, and $cov(e_{kj},e_{il})=0$ iff $k \neq i$. The stochastic non-negative weights of items are modeled as $n= K \cdot m$ random variables $\{w_{11},w_{12},\ldots,w_{1m},\ldots,w_{Km}\}$ where $w_{ij}$ denotes the weight of $j$-th item in group $i$. Item $e_{ij}$ has expected weight $E[w_{ij}]=a_{ij}$, variance $\sigma^2_{ij} = d$ and profit $p_{ij}$.

The chance-constrained knapsack problem with correlated uniform weights can be formulated as follows:
\begin{align}
     & \text{maximize} \  & p(x)= \sum_{i=1}^K \sum_{j=1}^m p_{ij}x_{ij} \\
     & \text{subject to} \ & \Pr(W(x) > B) \leq \alpha .
     \label{con:chance}
\end{align}
The objective of this problem is to select a set of items that maximizes profit subject to the chance constraint, which requires that the solution violates the constraint bound $B$ only with probability at most $\alpha$.

A solution is characterized as a vector of binary decision variables $x =(x_{11},x_{12},\\ \ldots,x_{1m},\ldots,x_{Km}) \in \{0,1\}^n$. When $x_{ij}=1$, the $j$-th item of the $i$-th group is selected. The weight of a solution $x$ is the random variable
\begin{align}
    W(x) =\sum_{i=1}^K \sum_{j=1}^m w_{ij}x_{ij},
\end{align}
with expectation
\begin{align}
    E[W(x)] =\sum_{i=1}^K \sum_{j=1}^m x_{ij},
\end{align}
and variance
\begin{align}
    Var[W(x)] = d\sum_{i=1}^K \sum_{j=1}^m x_{ij} + 2c\sum_{i=1}^K \sum_{1\leq j_1 < j_2\leq m}(x_{ij_1} x_{ij_2}). 
\end{align}

\begin{Definition}
Among all solutions with exactly $\ell$ one bits, we call a search point $x; |x|_1=\ell$ a \textbf{balanced solution}, denoted by $\ell^{b}$ if it selects $\left\lfloor\frac{\ell}{K}\right\rfloor$ items from $K-\left(\ell-\left\lfloor\frac{\ell}{K}\right\rfloor\cdot K\right)$ groups and $\left\lfloor\frac{\ell}{K}\right\rfloor +1$ items from the remaining $\ell-\left\lfloor\frac{\ell}{K}\right\rfloor\cdot K$ groups. This solution has covariance 
\begin{equation}
\resizebox{1.0\hsize}{!}{$s_{\ell}^{b}= c\left\{\left[K-\left(\ell-\left\lfloor\frac{\ell}{K}\right\rfloor\right)\right]\left\lfloor\frac{\ell}{K}\right\rfloor \left(\left\lfloor\frac{\ell}{K}\right\rfloor-1\right)+\left(\ell-\left\lfloor\frac{\ell}{K}\right\rfloor\right)\left(\left\lfloor\frac{\ell}{K}\right\rfloor+1\right)\left\lfloor\frac{\ell}{K}\right\rfloor\right\}.$} \nonumber
\end{equation}
Solutions with exactly $\ell$ bits that are not balanced solutions are called \textbf{unbalanced solutions}. 
\label{def:partition}
\end{Definition}

Among all unbalanced solutions, we call the following one the \textbf{most unbalanced solution} denoted by $\ell^{ub}$, which selects exactly $m$ items from $\left\lfloor\frac{\ell}{m}\right\rfloor$ groups and $\left(\ell-\left\lfloor\frac{\ell}{m}\right\rfloor \cdot m\right)$ items from another group. Since $m$ is the maximal number of items in each group, in the most unbalanced solution, there are $\left\lfloor\frac{\ell}{m}\right\rfloor$ full groups and one other group containing the remaining items. This solution has covariance $$s_{\ell}^{ub} = c\left[\left\lfloor\frac{\ell}{m}\right\rfloor m (m-1)+\left(\ell-\left\lfloor\frac{\ell}{m}\right\rfloor m \right)\left(i- \left\lfloor\frac{\ell}{m}\right\rfloor m -1\right)\right].$$ 

We calculate the upper bound of the covariance of acceptable solutions according to Chebyshev's inequality for all solutions with $\ell$ one bits. Denote by 
$$s_{\ell}= 2c \sum_{i=1}^K \sum_{1 \leq j_1 <j_2 \leq m} (x_{ij_1}  x_{ij_2})$$
the covariance of the solution $x$ and $\ell$ denotes the number of one bits in solutions. Then, the bound according to Chebyshev's inequality gives

\begin{align}
&\frac{\ell d+ s_{\ell}}{\ell d +s_{\ell} + (B-a\ell)^2}\leq \alpha \\
\Longleftrightarrow  & \ell d+s_{\ell} \leq \alpha({\ell}d+s_{\ell}+(B-a \ell)^2) \\
\Longleftrightarrow  & s_{\ell} \leq \frac{(B-a\ell)^2 \alpha}{1-\alpha}- \ell d.
\label{fun:coBound}
\end{align}
Therefore, the covariance of feasible solutions with exactly $\ell$ one bits is bounded above by
$\frac{(B-a \ell)^2 \alpha}{1-\alpha}-\ell d$.

In this paper, we assume the weights of items are correlated uniformly, and its hard to calculate the exact probability of violating the chance constraint. Similar to recent work on the uncorrelated problem~\cite{Yue19}, we use the one-sided Chebyshev's inequality (cf.\ Theorem \ref{thm:cheb}) to construct a usable surrogate of the chance constraint \eqref{con:chance}.

\begin{theorem}[One-sided Chebyshev's inequality]
\label{thm:cheb}
  Let $X$ be a random variable with expectation $E[X]$ and variance $Var[X]$. Then for any $k\in\mathbb{R}^+$, 
  \begin{equation}
    \Pr(X > E[X] +k)\leq \frac{Var[X]}{Var[X]+k^2}.
  \end{equation}
\end{theorem}

For the chance-constrained knapsack problem with correlated uniform weights, we define the surrogate function $\beta$ over decision vectors as 
\begin{align}
    \beta(x) =\frac{Var[W(x)]}{Var[W(x)]+(B-E[W(x)])^2}.
    \label{con:surrogate}
\end{align}
It is clear by Theorem \ref{thm:cheb} that $\Pr(W(x)\geq B) \leq \beta(x)$, and therefore every $x$ such that $\beta(x) \leq \alpha$ is also feasible. 

\begin{algorithm}[t]
\caption{Randomized Local Search (RLS)} 
\begin{algorithmic}[1]
\STATE Choose $x\in \{0,1\}^n$ to be a decision vector.
\WHILE { \textit{stopping criterion not met}}
\STATE Choose $b\in \{0,1\}$ randomly.
\IF {$b=0$}
\STATE{choose $i\in\{1,\ldots,n\}$ randomly and define $y$ by flipping the $i$th bit of $x$.}
\ELSE
\STATE{choose $(i,j)\in \{(k,l)|1\leq k<l\leq n\}$ randomly and define $y$ by flipping the $i$th and the $j$th bit of $x$.} 
\ENDIF
\IF{$f(y)\geq f(x)$} 
\STATE $x \leftarrow y$ ;
\ENDIF
\ENDWHILE 
\end{algorithmic}
\label{alg:rls}
\end{algorithm}

\begin{algorithm}[t]
\caption{(1+1)~EA}
\begin{algorithmic}[1]
\STATE Choose $x\in \{0,1\}^n$ to be a decision vector.
\WHILE { \textit{stopping criterion not met}}
\STATE $y\leftarrow$ flip each bit of $x$ independently with probability of $\frac{1}{n}$;
\IF{$f(y)\geq f(x)$} 
\STATE $x \leftarrow y$ ;
\ENDIF
\ENDWHILE 
\end{algorithmic}
\label{alg:oneone}
\end{algorithm}

We study the runtime of RLS and the (1+1)~EA defined in Algorithm \ref{alg:rls} and Algorithm \ref{alg:oneone} for optimization of the chance-constrained knapsack problem with dependent weights. RLS starts with a randomly initialized solution and iteratively improves it by applying a series of mutations. In each mutation step, it applies either \emph{one-} or \emph{two-bit} mutation with equal probability. Specifically, with probability $1/2$, it selects a single index uniformly at random from $\{1,\ldots,n\}$ an flips the corresponding bit in the current solution. Otherwise, it selects two distinct indexes uniformly at random to flip.
The (1+1)~EA also starts with a randomly initialized solution, but generates new candidate solutions by flipping each bit of the current solution with a probability of $1/n$, where $n$ is the length of the bit string. In the selection step, both algorithms accept the offspring if it is at least as good as the parent. We are interested in finding the optimal solution which is the feasible solution with maximal profit. We define the optimization time of RLS and the (1+1)~EA as the number of necessary steps until such an optimal solution is constructed. By considering the surrogate function obtained by the one-sided Chebyshev's inequality, we employ the fitness function
\begin{align}
    f(x)&:= (p'(x),\beta'(x)),
    \label{fun:fitness}
\end{align}
where $p'(x)=-1$ iff $\beta'(x) > \alpha$ and $p'(x)=p(x)$ otherwise, $\beta'(x) = \beta(x)$ iff $E[W(x)]< B$ and $\beta'(x) = 1+E[W(x)]-B$ otherwise. We optimize $f$ in lexicographic order where the goal is to maximize $p'(x)$ and minimize $\beta(x)$, i.e. we have
\begin{eqnarray}
& & f(x) \succeq f(y) \nonumber\\
&\Longleftrightarrow & p'(x) > p'(y)\\
& & \text{or} \left(p'(x)=p'(y)  \wedge \beta(x)\leq \beta(y)\right). \nonumber
\end{eqnarray}

Since selection is monotone, once a feasible solution is located, neither algorithm will subsequently accept an infeasible solution. Therefore, the process of finding an optimal solution can be separated into two parts: in the first part, the algorithm may first need to find a feasible solution. In the second part, it must find the highest profit among all feasible solutions.

\section{Obtaining feasible solutions}
\label{sec:reachfeasibel}

In this section, we analyze the expected time for RLS and the (1+1)~EA to find feasible solutions. 

\begin{Lemma}
Starting with an arbitrary initial solution, the expected time until RLS has obtained a feasible solution is $O(n\log n)$.
\label{lem:RLSfindfeasible}
\end{Lemma}

\begin{proof}
Adding a new item to the selected set will increase both the total expected weight and the probability of violating the chance constraint. Since all items have the same expected weight $a$, the sum of expected weight can be simply represented by the number of ones in the solution. 

The fitness function is defined in such a way that the total expected weight of a solution will never increase as long as no feasible solution has been obtained. This implies that RLS never accepts mutations that increase the number of ones, and only accepts a decrease in the number of ones. RLS cannot accept any single bit flips that flip a one to zero, or 2-bit flips that flip two one-bits to zeros.  

Therefore, at any solution $x; |x|_1=\ell$, there are $\ell$ one bits to decrease, and the probability to decrease the number of ones is at least $\frac{\ell}{2n}$. Hence, the expected waiting time until RLS constructs a feasible solution is bounded above by 
    $2n\left(1+ \cdots +\frac{1}{n}\right) = O(n\log n)$.
\end{proof}

\begin{Lemma}
Starting with an arbitrary initial solution, the expected time until the (1+1)~EA  obtains a feasible solution is $O(n^2\log n)$.
\label{lem:onereachfeasible}
\end{Lemma}

\begin{proof}
According to the definition of the fitness function in Equation \eqref{fun:fitness}, before finding a solution with expected weight less than $B$, the (1+1)~EA never accepts a solution that increases the number of one bits. Therefore, before producing such a solution, the algorithm only accepts mutations that reduce the number of one bits, and thus behaves identically to the optimization of the classical OneMax problem. The expected time for the (1+1)~EA to find a solution $x$ with $E[W(x)] < B$ is thus bounded by $O(n\log n)$, i.e., its expected running time on OneMax~\cite{Heinz92}. 

After finding a solution with expected weight less than $B$, the (1+1)~EA always accepts the solution with smaller constraint-violation probability according to the Chebyshev's inequality.
We construct a potential function $h :\{0,1\}^n$ as the sum of the variance and covariance of a solution,
\begin{align}
    h(x)= d \ell + 2c \sum_{i=1}^K  \sum_{1\leq j_1 < j_2\leq m}(x_{ij_i} x_{ij_2}), 
\end{align}
where $\ell$ denotes the number of items selected by solution $x$, $|x|_1 =\ell$ and $E[W(x)] < B$.

For a solution $x$ with $\ell$ one bits, the (1+1)~EA can reduce the potential $h(x)$ and the violation probability when flipping any one of the $\ell$ one bits to zero. Let the solution $y; |y|= \ell-1$ be an offspring generated from $x$ by flipping a one bit to zero. Then, we have $h(y)<h(x)$ and $\beta(y) < \beta(x)$. Let $x'$ be the next possible acceptable solution for the (1+1)~EA with $\ell$ one bits, then solution $x'$ should be better than solution $y$ according to the fitness function. We have $\beta(x') \leq \beta(y) < \beta(x)$ and $h(x')\leq h(y)<h(x)$ according to the Chebyshev's inequality. 

Given solution $x$, we consider all steps that flipping a 1-bit the algorithm generates solution $y$ after finding the solution $x'$ which reduces the value of the potential function. Let $\{r_1,\ldots,r_K\}$, where $0 \le r_i \le m$ denotes the number of items in group $i$ selected by $x$. Assume $x'$ is generated from $x$ by flipping a one bit from group $i$ to zero. Then, the reduction in potential for this mutation is the sum of the variance of this item and the difference of the covariance between $r_i$ elements and $r_i-1$ elements. That is,
$$d+ c(r_i(r_i-1)-(r_i-1)(r_i-2))=d+2c (r_i-1).$$ 
In group $i$, there are $r_i$ single bit flips that achieve this reduction, so the total contribution for group $i$ 
is $$dr_i+2cr_i(r_i-1)\geq dr_i+ 2c \frac{r_i(r_i-1)}{2}.$$ Summing over all groups we have 
\begin{align}
    \sum_{i=1}^K dr_i +2c r_i(r_i-1)\geq h(x).
    \label{fun:sumlowbond}
\end{align}

Therefore, after producing all single bit flips where each one bit of $\ell$ bits in $x$ has been flipped to zero once, the sum of gains with respect to the potential function should be as least as large as $h(x)$.

For all $t \in N $, let $x^{(t)}$ be the search point of the (1+1)~EA for the problem at time $t$ and $X^{(t)}=h(x)$. Then $$X^{(t)}-X^{(t+1)} = h(x^{(t)})- h(x^{(t+1)}).$$ Let $x\in \{0,1\}^n$ be a fixed nonempty solution, and let the points $y_1,\ldots,y_{\ell}$ be the $\ell$ different search points in $\{0,1\}^n$ generated from $x$ by first flipping one of the different $\ell$ one bit to zero. Thus, we have  by $h(y_i)\leq h(x)$ for all $i\in\{1,\ldots,\ell\}$ and inequality (\ref{fun:sumlowbond}) that
\begin{align}
    \sum_{i=1}^{\ell} \left(h(x)-h(y_i)\right) \geq h(x).
\end{align}
Since the $y_i$'s are generated from $x$ by a single bit flip each, we have 

\begin{align}
    \Pr(x^{(t+1)}=y_i\mid x^{(t)}=x)= \left( \left(1-\frac{1}{n}\right)^{(n-1)}\left(\frac{1}{n}\right)\right) \geq \frac{1}{en}
\end{align}
for all $i\in \{1,\ldots,\ell\}$. Furthermore
\begin{align}
    E[X^{(t)}-X^{(t+1)}\mid x^{(t)}=x, x^{(t+1)}=y_i] = h(x)-h(y_i)
\end{align}
holds for all $i\in\{1,\ldots,\ell\}$.

The (1+1)~EA never increases the current $h$-value of a search point, that is, $X^{(t)}-X^{(t+1)}$ is non-negative. Therefore, we have 
\begin{align}
    E[X^{(t)}-X^{(t+1)}\mid x^{(t)} =x] \geq \sum_{i=1}^k \left(h(x)-h(y_i)\right)\frac{1}{en}
\end{align}
and therefore, we have for all $x\in \{0,1\}^n$ that
$$E[X^{(t)}-X^{(t+1)}\mid x^{(t)}=x]\geq \frac{h(x)}{en}=\frac{X^{(t)}}{en}.$$

Therefore, the drift on $X^{(t)}$ is at least $\frac{h(x)}{en}$, and since the algorithm starts with $h(x)\leq h_{\ell}^{ub}$ where $h_{\ell}^{ub}$ denotes the sum of variance and covariance of the most unbalanced solution from the same level of $x$, and the minimum value of $h$ before reaching $h=0$ is 1, by multiplicative drift analysis, we find the expected time of at most
\begin{align}
    \frac{1+ \log(h_{\ell}^{ub})}{\frac{1}{en}} = O(n\log n)
\end{align}
to reach a solution with the number of one bits less than the starting search point. Let $\ell$ denote the number of one bits for the starting point, and the probability value of this solution is better than the best probability value for any solution with $\ell$ ones.

Furthermore, there is at most $n$ levels in the search space, and it takes $O(n\log n)$ steps for the (1+1)~EA to produce all possible solutions in each level. Altogether, the expected time of (1+1)~EA to search for a feasible solution is at most $O(n^2\log n)$.
\end{proof}

\section{Uniform Profits}
\label{sec:optimalSolution}

In this section, we assume the algorithms have produced a feasible solution and analyze the expected time that the (1+1)~EA and RLS require to find the optimal solution. We begin our study with the case that the deterministic profits are uniform. Since the actual value of profits does not affect the analysis, it is convenient to use unit profits.

\begin{Instance}
\label{Ins:uniform profit}

Given $K$ groups, each group has $m$ items. There are $n= K \cdot m$ items in total, the capacity of the knapsack is bounded by $B$. For $1\leq i \leq K$, $1 \leq j \leq m$, let $p_{ij}=1$, $a_{ij}=a$, $\sigma^2_{ij}=d$, where $d>0$ is a constant. The covariance of items within any group is $c$, i.e., we have $cov(e_{ij},e_{kl})=c$ iff $i=k$ and  $cov(e_{ij},e_{kl})=0$ otherwise.
\end{Instance}

\begin{Definition}
Let $r= \max\{|x|_1 \mid \exists x\in \{0,1\}^n \text{ with } \beta(x)\leq \alpha\}$ and partition the feasible search space by $L_0,L_1,\ldots,L_r$ such that 
\begin{align}
    L_i = \{x\in \{0,1\}^n : |x|_1 =i \ \text{ with} \ \beta(x)\leq \alpha\}.
\end{align}

We further bi-partition each partition $L_i$ into two sets $S_{i\gamma}$ and $S_{i\zeta}$ such that $S_{i\gamma} \cup S_{i\zeta} = L_i$ and $S_{i\gamma} \cap S_{i\zeta} = \emptyset$ as follows.

The set $S_{i\gamma} \subseteq L_i$ contains all feasible solutions where no extra item can be added without violating the chance constraint and $S_{i\zeta} \subseteq L_i$ is the set containing all feasible solutions where at least one extra item can be added to obtain a feasible solution with at least $i+1$ ones.
\label{def:bi-partition}
\end{Definition}

\begin{Lemma}
Starting with an arbitrary initial solution, the expected optimization time of RLS on the chance-constrained knapsack problem with correlated uniform weight is $O(n^3)$.
\end{Lemma}

\begin{proof}

Due to Lemma \ref{lem:RLSfindfeasible}, RLS finds a feasible solution in expected time $O(n\log n )$. Also, since all feasible solutions dominate any infeasible solution, the algorithm does not return to the infeasible region. 

Let $x\in L_{\ell}$. If $x\in S_{\ell \zeta}$, then there is at least one additional item that can be feasibly selected. This selection occurs with probability $1/2n$. Otherwise, only a 2-bit flip changing a zero and a one to a zero is accepted if it reduces the covariance of the solution without changing the profit until the algorithm produces a balanced solution on the same level. 

According to Definition \ref{def:partition}, the balanced solution in each level has the smallest covariance and number of items selected from each group. Let $l_i(x)$, $1\leq i\leq K$ be the number of elements chosen by $x$ from group $i$. Assume, without loss of generality, that the groups are sorted in increasing order with respect to the $l_i(x)$. Furthermore, let
\begin{align}
    s(x)=&\sum_{i=1}^{K-\left(\ell-\left\lfloor\frac{\ell}{K}\right\rfloor\cdot K\right)} \max\left\{0, \left\lfloor\frac{\ell}{K}\right\rfloor-l_i(x)\right\} \nonumber \\
    & +  \sum_{K-(\ell-\left\lfloor\frac{\ell}{K}\right\rfloor\cdot K)+1}^K \max\left\{0, \left\lfloor\frac{\ell}{K}\right\rfloor+1-l_i(x)\right\}
\end{align}
be the number of items that belong to an arbitrary balanced solution, but not chosen by $x$, and let
\begin{align}
    t(x)=&\sum_{i=K-\left(\ell-\left\lfloor\frac{\ell}{K}\right\rfloor\cdot K\right)+1}^K \max\left\{0, l_i(x)-\left(\left\lfloor\frac{\ell}{K}\right\rfloor+1\right)\right\} \nonumber \\
    & +  \sum_{i=1}^{K-\left(\ell-\left\lfloor\frac{\ell}{K}\right\rfloor\cdot K\right)} \max\left\{0, l_i(x)-\left\lfloor\frac{\ell}{K}\right\rfloor\right\}
\end{align}
be the number of items chosen by $x$, but do not belong to a balanced solution. Note that $s(x)$ should be equal to $t(x)$ for any feasible solution in $L_{\ell}$. Let $g=s(x)=t(x)$. 

As there are exactly $\ell$ 1-bits in the solution $x$, and $s(x)$ is a fixed value, this implies that there are at $s(x)$ 1-bits which can be swapped with an arbitrary 0-bit of the missing $g$ elements in order to reduce the covariance of $x$. Hence, the probability of such swapping is at least $g^2/2n^2$. Since $g$ cannot increase and $g\leq \ell$, it suffices to sum up these expected waiting times, and the expected time until reaching $g=0$ is
$$\sum_{g=1}^{\ell} (2n^2/g^2)= O\left(n^2(1-1/i)\right).$$ 

There are at most $n$ level of $L_{\ell}$ which implies that the expected time until an optimal solution has been achieved is 
$$\sum_{i=\ell-1}^n (n^2- n^2/i) = O(n^3-n^2\log n)= O(n^3),$$
which completes the proof.
\end{proof}

\begin{Lemma}
Let $x \in S_{\ell\gamma}$, then there exists some $q\in\{1,\ldots,n-1\}$ such that $q$ different accepted 2-bit flips of the (1+1)~EA reduce the covariance the solution. The expected one-step change of the (1+1)~EA is $\frac{X^{(t)}}{en^2}$.
\label{lemma:twobitflip}
\end{Lemma}

\begin{proof}
Let $x \in S_{\ell\gamma}, |x|_1=\ell$ and let $s_x$ denote the covariance of $x$. Then according to the inequality (\ref{fun:coBound}), $s_x$ is bounded by $\frac{(B-ai)^2 \alpha}{1-\alpha}-d\ell$. Let $x'\in S_{i\zeta}, |x'|_1=\ell$ be the balanced solution which take $\left\lceil\frac{\ell}{K}\right\rceil$ elements from the first $\ell-\left\lfloor\frac{\ell}{K}\right\rfloor K$ groups and take $\left\lfloor\frac{\ell}{K}\right\rfloor$ elements from the last $K-\ell-\left\lfloor\frac{\ell}{K}\right\rfloor K$ groups.

We assume that the $K$ groups in solution $x$ are sorted in increasing order with respect to the number of elements chosen by $x$. Then, let $q$ denote the Hamming distance between $x$ and $x'$, and $I=\max\{x-x',0\}$ denotes the set different of the elements chosen by $x$ but not by $x'$, and $I'=\max\{x'-x,0\}$ be the set different of element chosen by $x'$ but not by $x$. The number of elements in set $I$ and $I'$ should be the same and equal to $q$. A 2- bit flip flipping bit $i\in I$ from 1 to 0 and bit $j \in J$ from 0 to 1 can reduce the covariance of the problem which leads to the reduction of probability. As there are $q$ elements in set $I$ and $I'$ separately, they can be matched into $q$ pairs. Performing all such $q$ 2-bit flip can reduce the covariance of solution and simultaneously changes $x$ into $x'$.

Now, we have the reduction of covariance that flip $i\in I$ to zero denoted by 
$ 2c(r_i-1)$, where $r_i$ is the number of selected items of the group that $i$ belong to. There are $r_i-\left\lceil\frac{\ell}{K}\right\rceil$ one bits need to be flipped in this group that achieve this reduction to attend balance, and $r_i > \left\lceil\frac{\ell}{K}\right\rceil$, so the total contribution for this group is 
\begin{align}
    2c(r_i-1)\left(r_i-\left\lceil\frac{\ell}{K}\right\rceil\right) \geq c(r_i-1)r_i -c\left(\left\lceil\frac{\ell}{K}\right\rceil-1\right)\left\lceil\frac{\ell}{K}\right\rceil.
    \label{pot:reduction}
\end{align}

Similarly, we have the increase of covariance when flip a bit $j \in J$ denoted by $2c r_j$, where $r_j$ is the number of selected items of $x$ in the group that $j$ belong to. There are $\left\lfloor\frac{\ell}{K}\right\rfloor-r'_j$ zero bits flip in this group to attend balance and the total contribution for group $k'$ is 
\begin{align}
    2c(r'_j)\left(\left\lfloor\frac{\ell}{K}\right\rfloor-r'_j\right) \leq c\left(\left\lfloor\frac{\ell}{K}\right\rfloor-1\right)\left\lfloor\frac{\ell}{K}\right\rfloor - c(r'_j-1)r'_j.
    \label{pot:increast}
\end{align}
Define the total reduction of covariance from $x$ to $x'$ by inequalities (\ref{pot:reduction}) and (\ref{pot:increast}) as 
\begin{align}
   & \sum_{i=1}^q 2c(r_i-1) - \sum_{j=1}^q 2c r'_j \nonumber\\
   \geq & \resizebox{0.9\hsize}{!}{$c(r_i-1)r_i -c\left(\left\lfloor\frac{\ell}{K}\right\rfloor-1\right)\left\lfloor\frac{\ell}{K}\right\rfloor - \left(c\left(\left\lfloor\frac{\ell}{K}\right\rfloor-1\right)\left\lfloor\frac{\ell}{K}\right\rfloor - c(r'_j-1)r'_j\right)$} \nonumber\\
   =& s_x-s^b_{\ell}
\end{align}

Therefore performing all $q$ 2-bit flips simultaneously changes $x$ into $x'$ and leads to a covariance decrease at least as large as $s_x-s_{\ell}^b$, where $s^b_{\ell}$ denotes the covariance of the balanced solution with exactly $i$ one bits.

For all $t \in \mathbb{N}$, let $x^{(t)} \in L_i$ be a fixed, non-empty solution generated at time $t$ by the (1+1)~EA, and let $X^{(t)}= s_{x^{(t)}}-s_i^{b}$. Then
\begin{align}
    X^{(t)}-X^{(t+1)} = s_{x^{(t)}}-s_{x^{(t+1)}}.
\end{align}
Let $Y=\{y_1,\ldots,y_q\}$ with $q\in\{1,\ldots,n\}$ be the set of $q$ different search points that on the same level of $x$ in the search space generated from $x$ by one of the $q$ acceptable different 2-bit flips. We have $s_{y_i}\leq s_x$ for all $i\in \{1,\ldots,q\}$ and 
\begin{align}
    \sum_{i=1}^q (s_x - s_{y_i}) \geq s_x -s^b_i.
    \label{fun:sumdiffer}
\end{align}
Since each $y_i$ is generated from $x$ by one of the $q$ 2-bit flip,

\begin{align}
    \Pr[x^{(t+1)} \in Y | x^{(t)} =x] = q\left(1-\frac{1}{n}\right)^{n-2}\frac{1}{n}^2 \geq \frac{q}{en^2}
    \label{fun:probability11}
\end{align}
of the (1+1)~EA. Furthermore, 
\begin{align}
    E[X^{(t)}- X^{(t+1)} | x^{(t)} =x, x^{(t+1)} \in Y] = \frac{s_x-s^b_i}{q} = \frac{X^{(t)}}{q}.
    \label{fun:sumprob}
\end{align}

The algorithm cannot accept an offspring on the same level that increases the covariance, that is, $X^t -X^{t+1}$ is non-negative. Thus, we have by \eqref{fun:probability11} and \eqref{fun:sumprob} that
\begin{align}
    E[X^{(t)} - X^{(t+1)} | x^{(t)} =x] \geq \frac{X^{(t)}}{en^2}.
\end{align}

\end{proof}

\begin{Lemma}
\label{lem:one}
The expected time for the (1+1)~EA to transform a solution in $S_{i\gamma}$ to a solution in $S_{i\zeta}\cup L_j$ where $j>i$ is bounded by $O(n^2\log n)$.

\label{lem:transInlevel}
\end{Lemma}

\begin{proof}

According to Lemma \ref{lemma:twobitflip}, the drift on $X^{(t)}$ is at least $\frac{X^{(t)}}{en^2}$ for the (1+1)~EA. Therefore, since the both algorithms start with $X^{(t)}\leq s^i = O(n^2)$ and the minimum value of $X^{(t)}$ before reaching $X^{(t)}=0$ is 1, by multiplicative drift analysis, the expected time is at most $O(n^2\log n)$ to reach a solution in $S_{i\zeta}$. Then, if $i<r$, it is possible for the (1+1)~EA to generate a feasible in $L_{i+1}$ with probability $1/en$. The total expected time of the (1+1)~EA until an solution in $S_{i\zeta}\cup L_j$ is generated is thus bounded by $O(n^2\log n)$.
\end{proof}

\begin{theorem}
The expected time until the (1+1)~EA working on the fitness function \eqref{fun:fitness} constructs the optimal solution to Instance~\ref{Ins:uniform profit} is bounded by $O(n^3\log n)$.
\end{theorem}

\begin{proof}
By Lemmas \ref{lem:onereachfeasible} and \ref{lem:transInlevel}, for all $i<r$, it is sufficient to investigate the search process after having found a feasible solution $x\in S_{i\zeta}$, and after that, the algorithms can only accept an offspring with a larger number of one bits. It is possible for the (1+1)~EA to generate a feasible solution in $L_{i+1}$ by mutating exactly one zero bit to one. This event occurs with probability $\frac{1}{en}$ for the (1+1)~EA. 

Therefore, it will takes $O(n^2\log n+ en)$ steps to produce a feasible solution in level $L_{r+1}$ when started from a random feasible solution in $L_{r}$. Altogether, the expected optimization time is bounded by 
\begin{align}
    O(n^2\log n)+ \sum_{i=0}^{r-1} (n^2\log n + en)=O(n^3 \log n),
\end{align}
where $r<n$. 
\end{proof}

\section{Arbitrary profits mirrored by each group}
\label{sec:generalProfit}
We now turn our attention to the more complicated case where a single group has arbitrary profits, but this set of profits is the same for each of the $K$ groups. This resembles the case of general linear functions, but the chosen function is the shared by all groups.

\begin{Instance}
\label{Ins:Monotone profit}
Given $K$ groups, each group has $m$ items. There are $n= K \cdot m$ items in total, the capacity of knapsack is bounded by $B$. For $1\leq i \leq K$, $1 \leq j \leq m$, let $a_{ij}=a$, $\sigma^2_{ij}=d$ are constants, and let  $p_{i1}\geq p_{i2}\geq \ldots \geq p_{im}$ for $i\in\{1,\ldots,K\}$ and $p_{i\ell}=p_{k\ell}$ for each $i,k \in \{1, \ldots, K\}$, $1 \leq \ell \leq m$. The covariance of items in groups is $c$, i.e. we have $cov(e_{ij},e_{kl})=c$ iff $i=k$ and  $cov(e_{ij},e_{kl})=0$ otherwise.
\end{Instance}

\begin{theorem}
Starting with an arbitrary initial solution, the expected optimization time of RLS on the chance-constrained knapsack problem with correlated uniform weights is $O(n^3)$ on Instance~\ref{Ins:Monotone profit}.
\end{theorem}

\begin{proof}
By Lemma \ref{lem:RLSfindfeasible}, RLS finds a feasible solution in expected time $O(n\log n)$. Also, since all feasible solutions dominate infeasible solutions, the algorithm does not switch back to the infeasible region again. By the definition of Instance \ref{Ins:Monotone profit} that items in a group have different profit and the same weight, it is possible to have more than one balanced solution in each level of this case, but only one balanced solution with maximum profit, where we ignore the order of groups.

We order all items regarding to their profit as $p_{11} = p_{21} = \ldots = p_{K1}\geq p_{12} = p_{22}= \ldots = p_{K2}\geq \ldots \geq p_{1m}= p_{2m}= \ldots = p_{Km}$.

For a given solution $x$, we call the multi-set $P(x) = \{p_i \mid x_i=1\}$ the \emph{profit profile} of $x$, i.e., the multi-set of profit values selected by $x$.
We say that a profit profile $P$ is contained in $P(x)$ if $P \subseteq P(x)$.
Let $x$ be a feasible solution whose profit profile contains $P_j=\{p_1, \ldots, p_j\}$ (but which does not contain $P_{j+1}$).
We claim that RLS does not accept a solution whose profit profile does not contain $P_j$. An operation flipping a single $1$-bit that flips a 1 to 0 is clearly not accepted, as it reduces the profit and therefore cannot lead to a solution not containing $P_j$. A $2$-bit flip is only accepted if it does not decrease the profit, and therefore also cannot create a solution whose profit profile does not contain $P_j$, as $P_j$ contains the $j$-largest profits of the given input.

We analyze the time to transform a solution $x$ containing profit profile $P_j =\{p_i \mid 1 \leq i \leq j\}$ into a solution $x'$ containing profit profile $P_{j+1}$. Consider the profit $p_{j+1}$ in the group with the smallest number of elements whose bit $x_i$ is set to $0$. Flipping $x_i$ adds the profit $p_{j+1}$ to the profile $P_j$. 
Assume that bit $x_i$ belongs to group $r \in \{1, \ldots, K\}$, i.e., $x_i= x_{rs}$. If there is another item selected in group $r$ (selected by $x_{rs'}=1$) whose profit is less than $p_j$, then flipping both $x_{rs}$ and $x_{rs'}$ leads to an accepted solution $x'$ with $P_{j+1} \subseteq P(x')$. This happens with probability $1/2n^2$. Assume now that there is no such item in group $r$. Then $p_{j+1}$ is the largest non selected profit in group $r$. 

Let $S$ be the set of groups with the largest number of items selected and $\hat{p}$ the smallest selected profit from these groups. Assume that $x_i$ is not in one of the groups in $S$. Then flipping $x_i$ to $1$ and setting the bit corresponding to $\hat{p}$ to $0$ is accepted and leads to a solution containing profit profile $P_{j+1}$.
If $x_i$ is in one of the groups in $S$, then there is another item selected in $S$ with a profit smaller than $p_{j+1}$ or the solution $x$ is already optimal.

Altogether, to produce a solution $x'$ containing $P_{j+1}$ from a solution with $P_j$, RLS needs at most $O(n^2)$ steps, and since there are at most $n$ items in any solution, the expected optimization time of RLS is $O(n^3)$. 
\end{proof}

Let $p_{\max} = p_{i1}, i\in\{1,\ldots K\}$ be the maximal profit of the given input.
\begin{theorem}
Starting with an arbitrary initial solution, the expected optimization time of the (1+1)~EA on the chance constrained knapsack problem with correlated uniform weight is $O(n^3 (\log n + \log p_{\max}))$ on Instance~\ref{Ins:Monotone profit}.
\end{theorem}

\begin{proof}
According to Lemma \ref{lem:onereachfeasible}, the expected time to reach a feasible solution is $O(n^2\log n)$. Therefore it is sufficient to start the analysis with a feasible solution, after which the (1+1)~EA will never sample an infeasible solution during the remainder of the optimization process. 

For our analysis, we consider the set of all solutions $L_j= \{x \mid |x|_1=j; \beta(x) \leq \alpha\}$ containing exactly $j$ $1$-bits. For each $j$ we show that the expected number of offspring created from an individual in $L_j$ is $O(n^2 (\log n + \log p_{\max})$. After this many iterations, either the optimal solution (contained in $L_j$) has been created, or a feasible solution $y$ with $p(y) > \max_{x\in L_j} p(x)$ has been produced, which implies that the algorithm will not accept any solution in $L_j$ afterwards.

We now show that the expected number of offspring created from solutions in $L_j$ is $O(n^2 (\log n + \log p_{\max}))$.
Let $x \in L_j$ be the current solution, and let $x^{j, \max} = \arg \max_{x\in L_j} p(x)$ be an arbitrary feasible solution in $L_j$ with the largest possible profit.
Denote the \emph{loss} by
$$l(x) = \sum_{i=1}^n p_i x^{j,\max}_i (1-x_i),$$
that is, the sum of the profits chosen by $x^{j, \max}$ but not chosen by $x$. Denote the \emph{surplus} by
$$s(x) = \sum_{i=1}^n p_i (1-x^{j,\max}_i) x_i,$$
that is, the sum of the profits chosen by $x$ but not chosen by $x^{\max,j}$.
Define the total increase in profit from $x$ to $x^{j,\max}$ as
$$g(x) = p(x^{j, \max}) - p(x) = l(x) - s(x).$$ 

Let $r=\sum_{i=1}^n x^{j,\max}_i (1-x_i)$ be the number of indexes set to $1$ by $x^{j,\max}$ and $0$ by $x$. We give a set of $r$ accepted $2$-bit flips where the sum of the increases in profit is $g(x)$.

We consider the $K$ groups and w.l.o.g.\ assume that they are sorted in increasing order with respect to the number of elements chosen by $x =(x_{11},x_{12},\ldots,x_{1m},\ldots,x_{Km})$.
Let $\ell_i(x)$, $1 \leq i \leq K$ be the number of elements chosen by $x$ in group $i$.
We have $\ell_1(x) \leq \ldots \leq \ell_K(x)$. We consider the solution $\hat{x}^{j,\max}$ of maximal profit in $L_j$ for which 
$\ell_1(\hat{x}^{j,\max}) \leq \ldots \leq \ell_K(\hat{x}^{j,\max})$ and $\ell_K(\hat{x}^{j,\max}) \leq \ell_1(\hat{x}^{j,\max}) +1$ holds. This implies that $\hat{x}^{j,\max}$ is a balanced solution having the smallest variance in $L_j$.  
Note that such a solution exists as we may reorder the groups as each group contains the same (multi-)set of profits.

We have 
\begin{equation}
\label{eq:sizes}
  \sum_{i=1}^k \ell_i(x) \leq \sum_{i=1}^k \ell_i(\hat{x}^{j,\max}), 1\leq k \leq K  
\end{equation}
as both solutions contain $j$ elements and the groups are sorted in increasing order of the number of elements chosen by $x$.

This implies
\begin{equation}
\label{eq:sizes2}
\sum_{i=1}^k  \ell_i(x - \hat{x}^{j,\max}) \leq\sum_{i=1}^k \ell_i(\hat{x}^{j,\max} - x), 1\leq k \leq K  
\end{equation}
as the intersection of $\hat{x}^{j,\max}$ and $x$ contributes the same to each summand. Here $x - y =\max\{x-y,0\}$ denotes the set different of the elements chosen by $x$ but not by $y$.

Therefore, the number of elements chosen by $\hat{x}^{j,\max}$ but not by $x$ is greater than or equal to the number of elements chosen by $x$ but not $\hat{x}^{j,\max}$ for each of the first $k$ groups.

We then define our set of $r$ $2$-bit flips.
The $i$th $2$-bit flip flips the $i$th $0$-bit of $x$  (in the order given by the bit string $x =(x_{11},x_{12},\ldots,x_{1m},\ldots,x_{Km})$) set to $1$ in 
$\hat{x}^{j,\max} - x$ and the $i$th $1$-bit in $x - \hat{x}^{j,\max}$. 

Consider operation $i$ and let $p_i'$ be the profit introduced and $p_i''$ be the profit to be removed.
As per construction we have
$\sum_{i=1}^r p_i' = l(x)$ and
$\sum_{i=1}^r p_i'' = s(x)$ which implies that the total gain of the set of $r$ 2-bit flips is $g(x) = l(x) - s(x)$.
It remains to show that each of these $r$ $2$-bit flips is accepted by the algorithm.
Consider the $i$th operation. We show that $\beta(x)$ does not increase. Let $r'$ be the group that $p_i'$ belongs to and $r''$ be the group that $p_i''$ belongs to. We have $r' \leq r''$ due to Equation~\eqref{eq:sizes2} and therefore $\ell_{r'}(x) \leq \ell_{r''}(x)$. This implies that the $2$-bit flip leads to a solution $y$ with $\beta(y) \leq \beta(x)$.
We also have $p_i' \geq p_i''$ as otherwise we could improve the profit of $\hat{x}^{j,\max}$ which contradicts that $\hat{x}^{j,\max}$ is a feasible solution of maximal profit in $L_j$.

Given the set of $r$ accepted $2$-bit flips, the expected increase in profit is at least 
$$r/(en^2) \cdot g(x)/r = g(x)/(en^2),$$
as the probability of the (1+1)~EA to produce such a $2$-bit flip is $r/en^2$, and the average gain of this flip is $g(x)/r$.

For any non-maximal solution $x \in L_j$, we have $1 \leq g(x) \leq j\cdot p_{\max}$. Using multiplicative drift analysis, the expected number of offspring created from a solution $x\in L_j$ before having obtained a feasible solution $x^*$ with $p(x^*) \geq p(x^{j, \max})$ is therefore $O(n^2 (\log n + \log p_{\max})$. Moreover, $x^*$ is as better as the best solution in $L_j$, $x^*$ contains the top $j$ elements regarding to the profits of items, this implies that $x^*$ has the same construction as $x^{j, \max}$ and is a balanced solution that has smallest probability in $L_j$.

If $x^*$ is not optimal, then there exists a 1-bit flip adding an additional element and strictly improving the profit. There are at most $n$ level $L_j$ which implies that the expected time until an optimal solution has been achieved is $O(n^3 (\log n + \log p_{\max})$.
\end{proof}

\section{Conclusions}
The chance-constrained knapsack problem with correlated uniform weights plays a key role in situations where dependent stochastic components are involved. We have carried out a theoretical analysis on the expected optimization time of RLS and the (1+1)~EA on the chance-constrained knapsack problem with correlated uniform weights in this paper. We are interested in minimizing the probability that our solution will violate the constraint. We prove the bounds for both algorithms for producing a feasible solution. Then we carried out analyses of two settings for the problem, the one with uniform profits and the groups in the second case has the same profits profile. Our proofs are designed to provide insight into the structure of these problems and to reveal new challenges in deriving runtime bounds in the chance-constrained setting with the general type of stochastic variables. 

\section{Acknowledgements}
This research has been supported by the SA Government through the PRIF RCP Industry Consortium

\bibliographystyle{abbrv}
\bibliography{main}

\end{document}